\documentclass[a4paper,11pt]{article}

\usepackage[UKenglish]{babel}
\usepackage[T1]{fontenc}
\usepackage[utf8]{inputenc}

\usepackage{amsmath, amssymb, amsfonts}
\usepackage{nicefrac}
\usepackage{xspace}
\usepackage{color}
\usepackage{url}
\usepackage{bm}
\usepackage{bbm}
\usepackage[breakable]{tcolorbox}
\usepackage{algorithm}
\usepackage{algpseudocode}
\usepackage{blindtext}
\usepackage{multicol}
\usepackage{tikz,pgfplots}
\pgfplotsset{compat=1.18}
\usetikzlibrary{arrows}
\usepackage{multirow}
\usepackage{booktabs}
\usepackage[colorlinks=true,linkcolor=blue,citecolor=blue,urlcolor=blue]{hyperref}
\usepackage[capitalize,nameinlink]{cleveref}

\usepackage{geometry}
\usepackage{microtype}
\usepackage{amsthm}
\newtheorem{theorem}{Theorem}[section]
\newtheorem{lemma}[theorem]{Lemma}
\newtheorem{claim}[theorem]{Claim}

\theoremstyle{definition}
\newtheorem{definition}[theorem]{Definition}

\theoremstyle{remark}

\renewcommand{\nicefrac}[2]{\frac{#1}{#2}}

\title{Optimal Online Bipartite Matching in Degree-2 Graphs}

\author{
Amey Bhangale\thanks{University of California, Riverside, USA. \href{mailto:ameyb@ucr.edu}{ameyb@ucr.edu}. Supported by Hellman Fellowship award \& NSF CAREER award 2440882.}
\and
Arghya Chakraborty\thanks{Tata Institute of Fundamental Research, Mumbai, India. \href{mailto:arghya314@yahoo.com}{arghya314@yahoo.com}. Supported in part by Prof. Mrinal Kumar's SERB grant CRG/2023/006433. Research of the second and third authors supported by the Department of Atomic
Energy, Government of India, under Project Identification No. RTI4001.}
\and
Prahladh Harsha\thanks{Tata Institute of Fundamental Research, Mumbai, India. \href{mailto:prahladh@tifr.res.in}{prahladh@tifr.res.in}. Supported in part by the Google India Research Award.}
}

\date{\today}

\begin{document}
\maketitle
%TODO mandatory: add short abstract of the document
\begin{abstract}  
%\abnote{Placeholder}
Online bipartite matching is a classical problem in online algorithms and we know that both the deterministic fractional and randomized integral online matchings achieve the same competitive ratio of $1-\nicefrac1e$. In this work, we study classes of graphs where the online degree is restricted to $2$. As expected, one can achieve a competitive ratio of better than $1-\nicefrac1e$ in both the deterministic fractional and randomized integral cases, but surprisingly, these ratios are not the same. It was already known that for fractional matching, a $0.75$ competitive ratio algorithm is optimal. We show that the folklore \textsc{Half-Half} algorithm achieves a competitive ratio of $\eta \approx 0.717772\dots$ and more surprisingly, show that this is optimal by giving a matching lower-bound. This yields a separation between the two problems: deterministic fractional and randomized integral, showing that it is impossible to obtain a perfect rounding scheme.

\end{abstract}

%\newpage
\section{Introduction}

% \begin{itemize}
%     \item Bipartite matching
%     \item Online bipartite matching - edges can be unweighted or weighted.
%     \item Fraction vs integral matching.
%     \item Write about the history of the problem, including who started studying this problem and prior results, and conclude with the last result on edge weighted matching. 

%     \item Also, write about the tightness of the previous algorithms $(1-1/e)$. 
%     \item We can start introducing bounded online degree graphs. Here, you can start stating our results. (All 4 results, or two statements) 
    
%     \item After the integral matching theorem statement, 'this weird number turns out to be optimal'.
%     \item It would be good if we could describe the primal-dual algorithm informally (change in primal, change in dual, setting dual variables, satisfying dual constraints with a slack of $\gamma$, competitive ratio). Then we can say a few words about our algorithm for the integral case. 

%     \item After the fractional matching theorem statement, you should say that it is the same 'water-level' algorithm and we analyze its performance guarantee on $G(n,2)$ and also show that it is optimal.

%     \item Conclude with some open problems (for instance, optimal performance guarantee for online degree $d$ graphs)
% \end{itemize}

The problem of Bipartite Matching involves a bipartite graph $G=(L,R,E)$ with vertices $L\cup R$ and edges $E\subseteq L\times R$. The objective is to find the largest sized matching in $G$ - in other words, the largest subset $M\subseteq E$ such that no two edges in $M$ share a common vertex. This problem has been well studied \cite{Kuhn1955,Berge1957} and polynomial time algorithms are known.

Online Bipartite Matching was first introduced by Karp, Vazirani, and Vazirani~\cite{KarpVV1990}. In the online version of the problem, the set \( L \) is known in advance, but the vertices in \( R \) arrive one by one. Each time a vertex \( j \in R \) arrives (along with its edges), we must make an irreversible decision on whether to match \( j \) with one of its neighboring vertices \( i \in N(j) \subseteq L \). In a fractional version, we can assign fractional weights to the edges satisfying the matching constraints. The competitive ratio of an online matching algorithm is simply the approximation ratio achieved by the algorithm, that is, the worst-case ratio between the cost of the matching found by the algorithm and the cost of an optimal matching.  Online Bipartite Matching has found immense applications in real-world problems. AdWords, organ matching, and online dating are some of the few areas with direct applications. 

The seminal work of Karp, Vazirani, and Vazirani showed that if an unweighted graph has an optimal matching of size $n$ then no online algorithm can guarantee a matching of size larger than $(1-\nicefrac{1}{e})\cdot n$, in expectation. They also gave an algorithm, called RANKING, that attains the same competitive ratio of $(1-\nicefrac{1}{e})$ on unweighted graphs.

Many variants of this problem have also been studied. The survey by Mehta~\cite{Mehta2013} provides an excellent overview of the problem and its variants. Among the various generalizations studied, there is the edge arrival model \cite{GamlathKMSW2019}, where the edges arrive in an online fashion instead of the vertices. Recently, there have been breakthroughs~\cite{FahrbachHTZ2022,ShinA2021,GaoHHNYZ2021,BlancC2021} in the setting where the edges are weighted (one has to make additional assumptions here, as in this case, no competitive algorithm exists without an additional assumption). Other variants, like stochastic weights \cite{HaeuplerMZ2011,EzraFGT2020}, ad words \cite{MehtaSVV2007,GoelM2008,DevanurH2009}, display ads \cite{FeldmanMMM2009,CharlesCDJS2010,FeldmanHKMS2010}, online sub-modular welfare maximization \cite{LehmannLN2006,ChakrabartyG2010} have been studied. Most of these edge-weighted settings require either a stochastic input \cite{HaeuplerMZ2011,GoelM2008} or a free disposal model \cite{FeldmanKMMP2009,FahrbachHTZ2022,ShinA2021,BlancC2021}.

For all these matching-related problems, one may consider the polytope of matchings. Randomized algorithms naturally give rise to a fractional matching in this polytope. Hence, the problem of fractional matching is of extreme importance and has also been well studied. The WATER-LEVEL Algorithm \cite{CharlesCDJS2010} is a well-studied deterministic fractional matching algorithm, attaining a competitive ratio of $(1-\nicefrac{1}{e})$, and this is optimal. Almost all problems related to matching have this feature, which is that the optimal competitive ratio for randomized integral algorithms is the same as the optimal competitive ratio for fractional matching algorithms. Indeed, randomized integral algorithms inherently produce a probability distribution over matchings, and this probability distribution may be viewed as a fractional matching.

However, given that the competitive ratio is the same for all these related problems, it is interesting to ask if the fractional and the randomized integral matchings are inherently the same problems. Hence, people have tried to obtain integral matchings as a randomized rounding of fractional matchings \cite{BuchbinderNW2023}. In fact, for the edge weighted case, all the known algorithms \cite{FahrbachHTZ2022,ShinA2021,GaoHHNYZ2021,BlancC2021}, first produce a fractional matching and then use (randomized) Online Correlated Selection (OCS) algorithm to perform the task of online rounding. Therefore, it is interesting to ask whether the fractional and randomized integral matching problems are the same for other variants of matching. 

\noindent \textbf{Bounded degree graphs.} All the candidate graphs that prove the lower bound of $(1-\nicefrac{1}{e})$ competitive ratio for randomized integral or fractional matching have the feature that they contain very high-degree vertices. It is natural to ask whether one can construct low-degree graphs to obtain the same lower bounds. There have been some works on this \cite{AlbersS2022,albers_schubert_2022,CohenP2023,cohen_wajc,NaorW2018,BuchbinderJN2007,AzarCR2017}. In this paper, we study graphs where the online vertices have a degree at most $2$. We show that, indeed, algorithms can perform better in this particular setting. More specifically, we prove the following theorem.

\begin{theorem} [Integral Matching]
\label{thm:main_integral}
For any graph, $G$, with $n$ online vertices and degree of online vertices at most $2$, there exists an online randomized integral matching algorithm with competitive ratio $\eta$, and no online algorithm can attain a competitive ratio better than $\eta$,
\[\text{where }\eta := 1 - \sum_{i=1}^{\infty} \frac{1}{2^{{2^i} + i - 1}} \approx 0.717772.\]
\end{theorem}

This is interesting, in contrast to the following known result :

\begin{theorem} [Fractional Matching]
\label{thm:main_fractional}
For any graph, $G$, with $n$ online vertices and degree of online vertices at most $2$, the Water-Level Algorithm (c.f, \cref{water}) attains a competitive ratio of $0.75$ and no online algorithm can obtain a better competitive ratio.
\end{theorem}

The analysis of the fractional matching problem has been done in \cite{BuchbinderJN2007}, but for the sake of completeness, we provide an easy proof in \cref{sec:frac}.

It is intriguing that the fractional and the randomized integral matchings seem to be different problems for this class of graphs, as the optimal competitive ratio differs for the two different settings. 

We use the Primal-Dual analysis used in prior works \cite{BuchbinderN2009,FahrbachHTZ2022,BlancC2021} on online bipartite matching (described in detail in the following section) for our algorithm. Our main technical contribution is to set the primal-dual variables in the algorithm's analysis in a non-trivial manner, allowing us to prove the optimality of our algorithm.

\subsection{Brief overview}

In \Cref{sec:prelim}, we formally define some of the notations that we will use throughout the paper, the models that we look at, and the online primal-dual framework which is a commonly used tool in the analysis of online bipartite matching algorithms. Then in \Cref{lb}, we explicitly construct a graph, and then use Yao's principle to prove a lower bound for our problem. Finally, in \Cref{ub} we use the online primal-dual framework to prove the tightness of the above bound, by analyzing a simple algorithm.

\section{Preliminaries}\label{sec:prelim}

Throughout this paper, we shall denote by $G(d_1,d_2)$ the set of graphs that have degree at most $d_1$ for the offline vertices and degree at most $d_2$ for the online vertices. Note that a bipartite graph on $n$ online vertices and $n$ offline vertices can be any graph from $G(n,n)$. We shall focus on $G(n,2)$ graphs where the online vertices may have a degree at most $2$.

We shall use $L$ to refer to the set of offline vertices and $R$ to refer to the set of online vertices. To make things easier to read, we shall use $i,i_1,i_2$ to refer to vertices in $L$ (offline vertices) and $j$ to refer to vertices in $R$ (online vertices). We shall mostly consider an online vertex $j$, arriving with exactly two neighbors $i_1$ and $i_2$, and this is without loss of generality as given by the following lemma.

\begin{lemma}\label{exactdegree}
Online Bipartite Matching on graphs in $G(n,2)$ is equivalent to graphs with degree of online vertices exactly $2$. In other words, a $c$-competitive randomized algorithm for one implies a $c$-competitive randomized algorithm for the other.
\end{lemma}
\begin{proof}
Since \(G(n,2)\) includes graphs with online degrees at most 2, an online algorithm might perform better on graphs with exactly degree 2. We show that these two cases are equivalent. We only need to show that if we have a $\gamma$-competitive ratio algorithm on graphs with an online degree of exactly $2$, then there is a $\gamma$-competitive ratio algorithm on graphs with an online degree of at most $2$.

Fix an arbitrary algorithm $A'$ with $\gamma$-competitive ratio on graphs with an online degree of exactly $2$. From this, we will construct an algorithm $A$ with $\gamma$-competitive ratio on $G(n,2)$. Fix any graph \(G = (L,R,E) \in G(n,2)\). We construct a new graph \(G_m'=(L',R',E')\), parameterized by an integer $m$, where all online vertices have exactly degree $2$. This graph \(G_m'\) consists of \(m\) copies of \(G\), where the \(i\)-th copy is denoted by $G_i=(L_i,R_i,E_i)$ and an additional vertex $\ell^\star$ to adjust the degree of every online vertex in \(G'_m\) to exactly two. More precisely, the set of offline vertices \(L'\) is \(L_1\cup L_2\cup\dots\cup L_m\cup \{\ell^\star\}\) while the set of online vertices \(R'\) is \(R_1\cup R_2\cup\dots\cup R_m\). The edge set $E'$ is \(E_1\cup E_2\cup\dots \cup E_m\cup \tilde{E}\) where the edges in $\tilde{E}$ are defined next. 

The arrival sequence of the online vertices in $G'_m$ is as follows: all the $m$ copies of the first vertex in $R$, followed by all the $m$ copies of the second vertex in $R$, and so on. Recall, each of the copies $G_i=(L_i,R_i,E_i)$ are isomorphic to the original graph $G=(L,R,E)$. The graph \(G_m'\) has an additional set \(\tilde{E}\) of the edges as follows: For every $v\in R_1\cup R_2\cup\dots\cup R_m$ with degree $1$, we add an edge $(v, \ell^\star)$ to the graph $G'_m$.

Now, we run $A'$ on $G'_m$ for the first $m$ online vertices, followed by the next $m$ online vertices, and so on. The matching produced by $A'$ can be viewed as $m$ different matchings on $G$, with an additional vertex $\ell^\star\in L$. However, at most one of those matchings may contain $\ell^\star$ as part of the matching. The remaining $m-1$ matchings must not contain $\ell^\star$ and can be directly viewed as a matching on $G$. The algorithm $A$ will pick one of these $m$ matchings uniformly at random. If $A$ happens to pick the matching with $\ell^\star$ we give up and do not out put any matching. Note that $A$ can be generated from $A'$ in an online fashion, by selecting one of the $m$ matchings and simulating it. If this matching contains $\ell^\star$, we stop at that point and otherwise we output a feasible matching.

Now, let us bound the competitive ratio of $A$, assuming that the competitive ratio of $A'$ is $\gamma$. Let us assume that $G$ has an optimal matching of size $n'$, implying that $G'_m$ will have an optimal matching of size at least $m\cdot n'$. Hence $A'$, with a competitive ratio of $\gamma$, should give us at least $\gamma\cdot m\cdot n'$ size, in expectation. However, in this matching produced, we might have $\ell^\star$, which we cannot use in algorithm $A$. By our construction, we have $m$ collections of matchings of which one contains $\ell^\star$ and the remaining $m-1$ of them are used by $A$. Now, if we leave aside the particular matching that contains $\ell^\star$, we would still have a matching of size at least $\frac{\gamma\cdot m\cdot n'-n'}{m}$ in expectation, as produced by $A'$. The optimal matching in this set is of size at most $n'$. Hence the expected competitive ratio of $A$ is at least $\frac{\gamma\cdot m\cdot n'-n'}{m\cdot n'}=\frac{\gamma\cdot m-1}{m}$, which converges to $\gamma$ as $m$ increases. Hence, by choosing a large $m$, this conversion ensures that $A$ attains the same competitive ratio as $A'$ on $G$.

\end{proof}

\subsection{Primal and Dual programs for bipartite matching}
At this point, we would like to explain the primal and dual programs for online bipartite matching. We shall use $x_{i,j}$ to denote whether $i$ is matched to $j$ or not. So, we would want $x_{i,j}$ to be either $0$ or $1$, and the size of the matching produced would be the sum of $x_{i,j}$, which we would like to maximize. To ensure that it is a matching, we add the constraints that for each vertex, there must be at most one edge incident on the vertex.

We shall also use $x_i:=\sum_{j\in R}x_{i,j}$ to denote whether vertex $i$ is matched or not. We shall also relax this problem as follows :

\begin{minipage}{0.45\textwidth}
\vspace{10pt}
Maximize $\sum\limits_{i\in L}\sum\limits_{j\in R}x_{i,j}=\sum_{i\in L}x_i$

subject to :
$$\forall i\in L \sum\limits_{j\in \text{Nbd}(i)}x_{i,j}\leq 1$$
$$\forall j\in R \sum\limits_{i\in \text{Nbd}(j)}x_{i,j}\leq 1$$
$$x_{i,j}\geq 0$$
\end{minipage}\vline\hspace{10pt}\begin{minipage}{0.45\textwidth}
Minimize $\sum\limits_{i\in L}\alpha_i+\sum\limits_{j\in R}\beta_j$

subject to :
$$\forall (i,j)\in E, \alpha_i+\beta_j\geq 1$$
$$\forall i\in L, j\in R, \alpha_i\geq 0\text{ and }\beta_j\geq 0$$\vspace{1.1cm}
\end{minipage}

We have relaxed the $x_{i,j}$'s from being $\{0,1\}$ valued to instead allow any fraction :  $x_{i,j}\in[0,1]$. This means that the optimal value of the primal objective can only be larger now (since the graph is bipartite, the optimal value actually is the same), which means that the matching size obtained (the competitive ratio) will only be lower, compared to the primal optimal. Also note that if we are doing a fractional or a randomized matching, we would want $x_{i,j}$ to represent the fraction that $i$ is matched to $j$ or the probability of $i$ being matched to $j$, ensuring that the expected size of matching is $\sum_{i\in L,j\in R} x_{i,j}$ and the change/increment in the primal value because of the vertex $j$ is $\sum_{i\in L} x_{i,j} = \sum_{i\in L}\Delta x_{i}$.

The dual program represents the vertex cover problem, but we shall not use this fact in any way. Note the dual constraint $\alpha_i+\beta_j\geq 1$. If a dual solution is not feasible, meaning that the solution does not satisfy $\alpha_i+\beta_j\geq 1$ for all $(i,j)\in E$ but instead satisfies $\alpha_i+\beta_j\geq \gamma$ for some $\gamma<1$, then we shall call this solution $\gamma$-feasible.

In the analysis of our online algorithm, we will keep updating the primal and dual solutions. We shall use $\Delta P$ and $\Delta D$ to denote the increase in the primal and dual objectives, respectively. Similarly, we shall use $\Delta x_i$ and $\Delta \alpha_i$ to refer to the increments in the corresponding variables.

The main idea that we use in our algorithms is an online primal-dual method, similar to the one used in prior work~\cite{DevanurJK2013,FahrbachHTZ2022}. We always maintain feasible primal and approximately feasible dual solutions. Whenever primal increases by some amount, we increase the dual by the same amount, so the primal and dual always have the same value. The primal will be feasible, so we would obtain an optimal solution if the dual were also feasible. However, we will be unable to maintain a feasible dual solution; instead, we will maintain an approximate feasible solution that satisfies the dual constraints with a slack of $\gamma$, in other words, be $\gamma$-feasible.

More specifically, we will always maintain the following throughout :

\begin{itemize}
    \item \textbf{Approximate Dual Feasibility :} For any $i\in L$ and $j\in R$, $\alpha_i+\beta_j\geq \gamma$
    \item \textbf{Reverse Weak Duality :} The objective value of the primal (P) and the objective value of the dual (D) satisfy that $P\geq D$.
\end{itemize}

This is important because of the following claim.
\begin{claim}
If the primal solution is feasible and has value $P$, and if the dual solution has objective value at most $P$, while being $\gamma$-feasible, then the primal value $P$ is at least $\gamma$ fraction of the primal optimal.
\end{claim}
\begin{proof}
Let $M$ be the primal optimal. By strong duality, it is also the value of the dual optimal. Since the primal is feasible, $P\leq M$. Let the dual solution have value $D\leq P$. Since the dual solution is $\gamma$ feasible, multiplying all the dual variables by $\nicefrac{1}{\gamma}$ gives a feasible dual solution with value $\nicefrac{D}{\gamma}\geq M$ because it is feasible.

Hence $D\leq P\leq M\leq \nicefrac{D}{\gamma}\implies \gamma\cdot M\leq D\leq P\implies\nicefrac{P}{M}\geq \gamma$.
\end{proof}

\section{Techniques}

While the algorithm that we suggest is extremely simple, the interesting part is the analysis. We use the online primal-dual method to prove that the algorithm is optimal. The primal variables are easy to update, given the algorithm. The tricky part is to update the dual variables: $\alpha_i$ and $\beta_j$. Indeed, there are other intuitive candidate algorithms for this problem, which we tried but were unable to analyse using the primal-dual approach.

On arrival of any vertex, $j\in R$, we try to increase the $\alpha_i$'s as much as possible and $\beta_j$ as little as possible while maintaining $\gamma$-competitive ratio. The key observation, in  \Cref{lemma:setting_alpha}, is that given the primal variable of a vertex $i\in L$, one can determine a lower bound for the dual variable $\alpha_i$. Given this observation, it is sufficient to show that as the probability of $i$ being matched approaches $1$, the dual variable $\alpha_i$ approaches $\gamma$.

\section{Randomized Integral Matching}
We start with a lower-bound construction first. 
\subsection{Lower Bound}\label{lb}
Recall that $\eta := 1-\sum_{i=1}^\infty \frac{1}{2^{2^i+i-1}}$. The following theorem shows that no online algorithm can achieve a competitive ratio better than $\eta$.
\begin{theorem}
For Online Integral Bipartite Matching, no online algorithm (deterministic or randomized) can attain a competitive ratio better than $\eta$.
\end{theorem}
\begin{proof} We will use Yao's principle to prove this. We shall first construct one particular graph $G$ and then look at the uniform distribution over all graphs that are isomorphic to the described graph.

\begin{figure}[ht]
    \centering
    \includegraphics[scale=0.3]{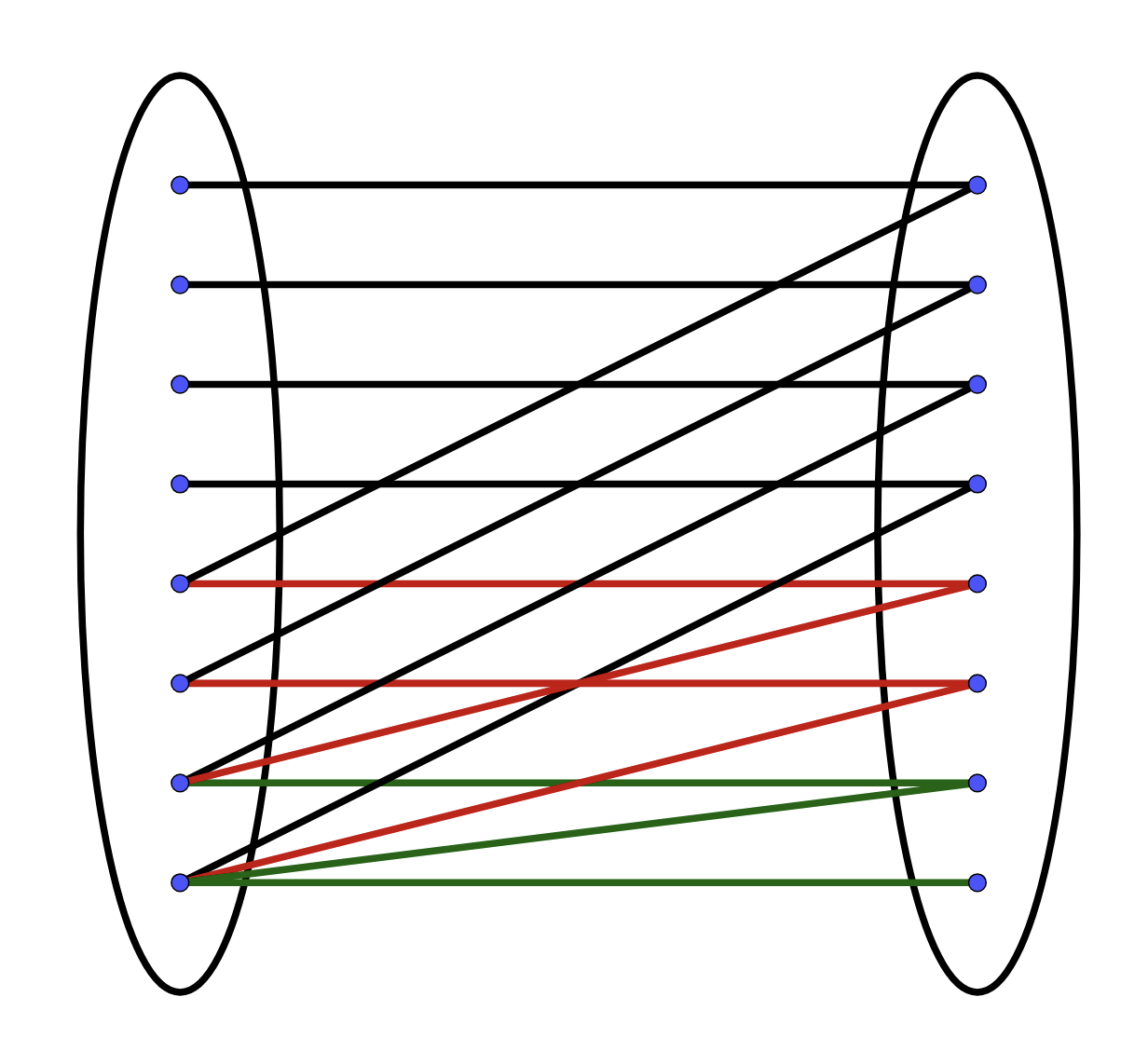}
    \caption{The graph for the lower bound, with $8$ online and offline vertices with online degree bounded by $2$. The black edges correspond to online vertices in the first phase, the red ones correspond to the second phase and the green ones correspond to the third phase.}
    \label{fig:example}
\end{figure}

This graph will have $n=2^k$ online and $2^k$ offline vertices for some integer $k$. We will describe the graph in phases. There will be $k=\log n$ phases in total. The online $i$-th vertex will always have an edge to the offline $i$-th vertex. Since the online vertices will have degree $2$, it suffices to explain the other neighbor of the online vertex $i$. \Cref{fig:example} depicts this graph for $k=3$, with $2^k=8$ online and offline vertices.

This is where we will use the phases to describe the neighbors. In the first phase, there will be $\nicefrac{n}{2}$ online vertices: $1$ through $\nicefrac{n}{2}$ and the $i$-th vertex would have the other edge to offline vertex $\nicefrac{n}{2}+i$. In the second phase, there will be $\nicefrac{n}{4}$ online vertices: $\nicefrac{n}{2}+1$ through $n-\nicefrac{n}{4}$ and the $\nicefrac{n}{2}+i$-th vertex would have the other edge to offline vertex $n-\nicefrac{n}{4}+i$. In the $j$-th phase, there will be $\nicefrac{n}{2^j}$ online vertices: $n-\nicefrac{n}{2^{j-1}}+1$ through $n-\nicefrac{n}{2^j}$ and the $n-\nicefrac{n}{2^{j-1}}+i$-th vertex would have the other edge to offline vertex $n-\nicefrac{n}{2^j}+i$.

Notice that in the $j$-th phase, when an online vertex arrives, its neighbors' degree changes from $j-1$ to $j$.

Now using this graph, $G$, let us construct a distribution over graphs. We shall consider a set of all graphs that can be obtained from $G$, by a permutation of the offline vertices, and then take a uniform distribution over these graphs. The idea is to use Yao's principle on this distribution and study the competitive ratio of deterministic algorithms. We shall compute the probability of offline vertices being matched, based on the phases. We shall consider any candidate deterministic algorithm for this.

Consider any offline vertex, $u$. If there is an online vertex, say $v$, in the $j$-th phase, with $u$ as a neighbor of $j$, then we shall say that $u$ ``appeared'' in the $j$-th phase. Notice that if $u$ appeared in the $j$-phase, then $u$ must have appeared in the phases $1,2,\dots,j-1$ too. For every offline vertex, $u$, there will be a unique $j$ such that $u$ appeared in the $j$-th phase and then did not appear in any of the phases $j+1,j+2,\dots$ again.

Let us first look at offline vertices that appear in the first phase and then never again. The probability that such a vertex is matched is $\nicefrac{1}{2}$, and there are $\nicefrac{n}{2}$ such vertices.

Next, for vertices that appear in the second phase and never after that, we shall compute the probability that they are not matched. Let us consider one such vertex, $u$. For $u$ not to be matched, it has to be necessarily unmatched when it appeared in the first phase. This would happen with probability $\nicefrac{1}{2}$. Also, when $u$ appears in the second phase, the online vertex, say $v$, will have another neighbor, say $i_1$. If $i_1$ happened to be matched earlier, in phase 1, then $u$ will get matched in the second phase. Therefore, for $u$ to be unmatched $i_1$ must be unmatched and then on the arrival of $v$ the probability that the deterministic algorithm will leave $u$ unmatched will be $\nicefrac{1}{2}$. Hence the probability that $u$ is unmatched is $\nicefrac{1}{8}$.

For a vertex $u$ that appears in the third phase and never thereafter, the probability that $u$ is unmatched is $\nicefrac{1}{2}\cdot\nicefrac{1}{8}\cdot\nicefrac{1}{8}$ because this is the case when $u$ was unmatched even after phase two and $i_1$, the competitor of $u$ in phase three was also unmatched after phase two and in this case the probability of $u$ being unmatched is $\nicefrac{1}{2}$.

Similarly, for a vertex that appears in the $i$-th phase and then never thereafter, the probability that such a vertex is unmatched is $\nicefrac{1}{2^{2^i-1}}$. Also, there are $\frac{n}{2^i}$ such vertices that appear in the $i$-th phase and then never thereafter.

Hence, the expected number of vertices unmatched equals $$\sum_{i=1}^\infty \frac{n}{2^{2^i+i-1}}.$$
Therefore, the competitive ratio will be $$1-\sum_{i=1}^\infty \frac{1}{2^{2^i+i-1}}=\eta\approx 0.717772.$$
\end{proof}
\subsection{Upper Bound}\label{ub}
In this section, we will show that the folklore \textsc{half-half} algorithm attains a competitive ratio of \(\eta\). We begin by recalling the \textsc{half-half} algorithm, which can be described as follows. When an online vertex $v$ comes, if it has only one unmatched neighbor, then match $v$ with the unmatched neighbor. If both the neighbors of $v$ are unmatched, then pick one uniformly at random and match it with $v$. The formal description of the algorithm is given in \Cref{alg}.

\begin{algorithm}[H]
  \caption{Half-Half Algorithm}\label{alg}
  \begin{algorithmic}[1]
    \Procedure{$A$}{$j,S$}\Comment{Online vertex $j$ with neighborhood $S\subseteq L$, $|S|\leq 2$}
      \If {$\forall i\in S$, $i$ is already matched}
      \State Do nothing.
      \Else
      \If {There exists only one $i\in S$, such that $i$ is unmatched}
      \State Match $i$ to $j$.
      \Else
      \State We have $S=\{i_1,i_2\}$ such that $i_1$ and $i_2$ are both unmatched
      \State Match $j$ with $i_1$ and $i_2$ with probability half each.
      \EndIf
      \EndIf
    \EndProcedure
  \end{algorithmic}
\end{algorithm}

%The analysis of the 2/3 competitive ratio moved to old-calculations.tex

\subsubsection{Analysis of the half-half algorithm}

\begin{theorem}
\label{thm:integral_eta}
The half-half algorithm attains a competitive ratio of $\eta$, where :
\[
\eta := 1 - \sum_{i=1}^{\infty} \frac{1}{2^{{2^i} + i - 1}}
= 1 - \frac{1}{4} - \frac{1}{32} - \frac{1}{2^{10}} - \frac{1}{2^{19}} - \dots \approx 0.717772.
\]
\end{theorem}
Throughout the rest of the section, we are going to prove this theorem.

The plan is again to use the online Primal-Dual analysis. We will show that on the arrival of $j\in R$, with neighbors $i_1,i_2\in L$, we can update the primal such that the expected size of the matching produced is at least the value of the primal. Simultaneously we will set up $\beta_j$ and increment $\alpha_{i_1},\alpha_{i_2}$ so that the dual conditions are $\eta$-feasible. We shall initialize by setting all the $\alpha_i$'s to $0$, and $\beta_j$ will start off as $0$ before we do the increments for the dual variables. We shall describe how to update the dual variables later. First, we will describe the updates to the primal variables.

\noindent{\bf Updating the Primal variables.} According to the algorithm, we can change the primal variables according to the following rule. We always maintain the primal variables to take values of the form $1-\frac{1}{2^p}$ where $p \in \{0,1,2,3,\ldots\}$. If an online vertex $j$ comes with neighbors $\{i_1, i_2\}$, and the old values of $x_{i_1}$ and $x_{i_2}$ are as follows:

\[
x_{i_1} = 1 - \frac{1}{2^{p_1}}, \ \ \ \ \ x_{i_2} = 1 - \frac{1}{2^{p_2}},
\]
then the updated values would be
\[
x_{i_1} = x_{i_2} = 1 - \frac{1}{2^{p_1+p_2+1}}.
\]
With this, the change in the primal value would be

\[
\Delta P = 2\left(1-\frac{1}{2^{p_1+p_2+1}}\right) - \left(1-\frac{1}{2^{p_1}}\right) - \left(1-\frac{1}{2^{p_2}}\right).
\]

\begin{claim}\label{pr_values}
Changing the primal in this way ensures that primal variables are always of the form $1-\frac{1}{2^p}$, and that the expected increment of matching size is at least the change in primal.
\end{claim}

\begin{proof}
First of all, it is clear from these updates that the primal variables are always of the form $1-\frac{1}{2^p}$. Had the events of $i_1$ and $i_2$ being unmatched been independent events, we could have said 
\[\mathbb{P}[i_1\text{ is unmatched after }j]=\mathbb{P}[i_2\text{ is unmatched after }j]\]\[=\frac{1}{2}\cdot \mathbb{P}[i_1\text{ and }i_2\text{ were unmatched before }j]=\frac{1}{2^{p_1+p_2+1}}.\]

However, if the events are not independent, then there is only a negative correlation, i.e.,  given that $i_1$ is unmatched, the probability of $i_2$ being unmatched can only decrease. The proof is as follows. Conditioned on $i_1$ being unmatched, every online vertex, with $i_1$ as its neighbor, must have been matched to the offline vertex other than $i_1$. Let us call this set of vertices $S$. Now, when these offline vertices from $S$, after being matched, were neighbors of some online vertex, say $v$, the online vertex $v$ must have been matched to the other neighbor of $v$. We can update our $S$ to include all such vertices that get matched in this manner. Continuing like this, the set $S$ will consist of all vertices that are matched, conditioned on $i_1$ not being matched. If any online vertex had $i_2$ as its neighbor and a vertex from $S$ as its other neighbor, $i_2$ must have been matched conditioned on $i_1$ not being matched. Otherwise, $i_1$ being matched and $i_2$ being matched are independent events. Hence,

\[\mathbb{P}[i_1\text{ is unmatched after }j]=\mathbb{P}[i_2\text{ is unmatched after }j]\]\[=\frac{1}{2}\mathbb{P}[i_1\text{ and }i_2\text{ were unmatched before }j]\leq\frac{1}{2^{p_1+p_2+1}}.\]

We have already seen that the change in primal is
\[
\Delta P = 2\left(1-\frac{1}{2^{p_1+p_2+1}}\right) - \left(1-\frac{1}{2^{p_1}}\right) - \left(1-\frac{1}{2^{p_2}}\right)=\frac{1}{2^{p_1}}+\frac{1}{2^{p_2}}-\frac{2}{2^{p_1+p_2+1}}.
\]

Now we can compare this to the expected increase in matching size,
\begin{align*}
\mathbb{E}[&\text{Increment in size of matching}]\\
&= \mathbb{P}[\text{Either }i_1 \text{ or }i_2\text{ was unmatched on arrival of } j]\\&=\mathbb{P}[i_1\text{ was unmatched on arrival of } j]+\mathbb{P}[i_2\text{ was unmatched on arrival of } j]-\\&\qquad\qquad\qquad\qquad\qquad\qquad \mathbb{P}[\text{Both }i_1 \text{ and }i_2\text{ were unmatched on arrival of } j]\\&\geq \frac{1}{2^{p_1}}+\frac{1}{2^{p_2}}-\frac{2}{2^{p_1+p_2+1}}.
\end{align*}
Thus,
$$ \mathbb{E}[\text{Increment in size of matching}] \geq \Delta P ,$$
and this proves the claim.
\end{proof}

\noindent{\bf Updating the Dual variables.} Now that we have described how the primal variables are updated, we shall describe how to update the dual variables. As mentioned earlier, we will always maintain $\Delta P\geq \Delta D$. So, when online vertex $j$ with neighbors $i_1,i_2$ arrives, we shall first compute $\Delta P$ and then try to set up $\alpha_{i_1}=\alpha_{i_2}$. The goal is to set up the $\alpha$'s as high as possible, preferably to $\eta$, and assign the remainder of $\Delta P$ to $\beta_j$ if needed. This will ensure $\Delta P\geq\Delta D$. We would want 
$\alpha_{i_1}+\beta_j\geq \eta$ and also $\alpha_{i_2}+\beta_j\geq \eta$, while simultaneously maintaining $\beta_j\geq 0$. This would limit us from setting large $\alpha$'s. It will be convenient to define the following notation.

\begin{definition} [$\alpha_{(k)}$]
    We define the quantity $\alpha_{(k)}$ as follows: On arrival of the online vertex $j$ with neighbors $i_1$ and $i_2$, suppose the primal variables $x_{i_1}$ and $x_{i_2}$ take the value $1-\frac{1}{2^k}$ {\em after} the update, then the minimum value that the dual variables $\alpha_{i_1}$ and $\alpha_{i_2}$ are set (while maintaining $\eta$-feasibility and $\beta_j\in [0,1]$) is the quantity $\alpha_{(k)}$.
\end{definition}

What we will show now is that it will always be possible for us to set the dual variables in order to maintain $\eta$-approximate dual feasibility.

We shall first show that for $1\leq i\leq 3$, we can set up $\alpha_{(i)}$ to maintain approximate dual feasibility. After that, we will prove a general statement showing that we can assign appropriate dual variables for all $x_i$.

\vspace{10pt}
\noindent{\bf Setting $\alpha_{(1)}$}. Let us first consider the case when the two neighbors $i_1,i_2$ had not appeared before. In this case, the half-half algorithm will set $x_{i_1}=x_{i_2}=\nicefrac{1}{2}$. Observe that this is the only case when the $x_i$ will be set to $\nicefrac{1}{2}$, and hence the previous values of $x_{i_1},x_{i_2}$ is $0$. So, $\Delta P=1$. Let us set $\alpha_{i_1}=\alpha_{i_2}= \alpha_{(1)} = 1-\eta$ in this case, and $\beta_j$ to $2\eta-1$. This will ensure that $\Delta D= 1=\Delta P$ (recall, $\alpha_{i_1}$ and $\alpha_{i_2}$ were set to $0$ at the beginning). Additionally, the dual constraints for the edges $\{i_1,j\},\{i_2,j\}$ are $\eta$-satisfied.

\vspace{10pt}
\noindent{\bf Setting $\alpha_{(2)}$}. Let the neighbors of $j$ be $i_1$ and $i_2$ with $x_{i_1}=1-\nicefrac{1}{2}$ and $x_{i_2}=0$. Then the updated $x_{i_1}=x_{i_2}=1-\nicefrac{1}{4}$ and notice that, up to symmetry, this is the only way for the primal variables to attain $x_i=1-\nicefrac{1}{4}$. In this case, $\Delta P=1$. Updating $\alpha_{i_1}=\alpha_{i_2}=\alpha_{(2)} = 2-2\eta$ means that the increment in $\alpha_{i_1}+\alpha_{i_2}$ is $(\alpha_{(2)} -\alpha_{(1)})+ (\alpha_{(2)} - 0) = ((2-2\eta)-(1-\eta))+((2-2\eta)-0)$ because for $i_1$, with $x_{i_1}=0$, we had $\alpha_{i_1}=1-\eta$ and for $i_2$ we had $\alpha_{i_2}=0$ before the arrival of $j$. Hence the total increment in $\alpha_{i_1}+\alpha_{i_2}$ is $3-3\eta$. Therefore, assigning $\beta_j$ to $\Delta P-\Delta\alpha_{i_1}-\Delta\alpha_{i_2}=3\eta-2$ ensures that $\beta_j$ is positive and the increment in primal is equal to the increment in the dual. Also, for the edges, the dual constraints are $\eta$-satisfied since $(2-2\eta)+(3\eta-2)=\eta$.

\vspace{10pt}
\noindent{\bf Setting $\alpha_{(3)}$}. There are two cases in this case.
\begin{enumerate}
    \item In the first case, let the neighbors of $j$ be $i_1$ and $i_2$ with $x_{i_1}=x_{i_2} = 1-\nicefrac{1}{2}$ before the arrival of $j$. According the the change in primal values, both $x_{i_1}$ and $x_{i_2}$ would be set to $1-\nicefrac{1}{8}$ after the update. Therefore, $\Delta P=2(1-\nicefrac{1}{8})-\nicefrac{1}{2}-\nicefrac{1}{2}=\nicefrac{3}{4}$. Also, $\alpha_{i_1}$ and $\alpha_{i_2}$ used to be $\alpha_{(1)}=1-\eta$ before $j$ arrived. If we update $\alpha_{i_1}$ and $\alpha_{i_2}$ to $\alpha_{(3)}$, then we can set $\beta_j$ to $\Delta P-(2\alpha_{(3)}-2(1-\eta))$ in order to ensure that the change in dual is at most the change in primal. However, after the update, the edge corresponding to $\{i_1,j\}$ must satisfy approximate dual feasibility. Hence, we must have $\beta_j+\alpha_{(3)}\geq \eta$, implying $\nicefrac{3}{4}-\alpha_{(3)}+2-2\eta\geq\eta$. Therefore $\alpha_{(3)}\leq \nicefrac{11}{4}-3\eta$.

    \item In the second case, if the two neighbors $i_1$ and $i_2$ had $x_{i_1}=1-\nicefrac{1}{4}$ and $x_{i_2}=0$ (i.e. $i_2$ not having appeared before) even then both $x_{i_1}$ and $x_{i_2}$ would be set to $1-\nicefrac{1}{8}$ after the update. It can be checked that in this case, $\Delta P=1$ and hence we can set $\alpha_{i_1}=\alpha_{i_2}=\eta$, while keeping $\beta_j=0$.
\end{enumerate}

Considering both these cases, we can ensure that $\alpha_{i_1}, \alpha_{i_2}$ will certainly be set to at least $\nicefrac{11}{4}-3\eta$ and hence $\alpha_{(3)} = \nicefrac{11}{4}-3\eta$.

\noindent{\bf Setting $\alpha_{(k)}$ for $k\geq 4$}. As we have seen till now, we can set up the $\alpha$ variables so that approximate dual feasibility is satisfied.

The next lemma gives a strategy for the remaining cases to set the dual variables so that the Dual is $\eta$-feasible. 
\begin{lemma}
\label{lemma:setting_alpha}
For every $k\geq 4$, there is a fixed $\alpha_{(k)}$ such that whenever the primal variables are set to $1-\frac{1}{2^k}$, one can set the value of the corresponding dual variables $\alpha$ such that $\alpha\geq \alpha_{(k)}$ while ensuring $\Delta D\leq \Delta P$, $\beta_j\in [0,1]$ and the $\eta$-approximate dual feasibility. More specifically, the algorithm sets $\alpha_{(k)}$ to be $\eta$ whenever $k$ is not of the form $2^m-1$ (for some integer $m$), otherwise $$\alpha_{(2^m-1)}\geq\sum_{i=1}^{m-1} 2^{m-1-i} \left( \frac{1}{2^{2^i-2}} - \frac{1}{2^{2(2^i-2) + 2}}\right) + 2^{m-1} - (2^m -1) \eta.$$
\end{lemma}

Before we prove this lemma, let's first see why this is enough to prove the main theorem.

\begin{proof} [Proof of \Cref{thm:integral_eta}]
It is easy to see that proving this lemma completes the proof. This is because, after the arrival of $j$, $i_1$ and $i_2$ are more likely to be matched and hence $x_{i_1}$ and $x_{i_2}$ increase. We have seen that $x_i$ may only take values of the form $1-\nicefrac{1}{2^k}$ for some integer $k$. For $k=1$ to $3$, we have explicitly checked that we can always satisfy reverse weak duality and $\eta$-approximate dual feasibility. For any larger value of $k$, the \Cref{lemma:setting_alpha} tells us that we will always be able to set up the dual variables while satisfying reverse weak duality and $\eta$-approximate deal feasibility. Hence, we shall always be able to update the dual variables in this manner, for all $k$ and therefore the Half-half Algorithm would be $\eta$-competitive.
\end{proof}

We now prove \cref{lemma:setting_alpha}.
\begin{proof}[Proof of \Cref{lemma:setting_alpha}]
The proof will be by induction on $k$. Let us assume that $\alpha_{(k)}$ follows the lemma up to some $k<T$.

We will need to check the base cases for $k=4,5,6,7$ after which the proof follows by induction. The calculations are exactly like that for $k=1,2,3$. We provide a table of values here for each $k$. This table describes how the dual variables are set for that particular $k$. Each row describes a different initial setting of $x_{i_1}$ and $x_{i_2}$ that can lead to $x_{i_1}=x_{i_2}=1-\nicefrac{1}{2^k}$. The reader may check that these updates satisfy $\Delta P\geq \Delta D$ and $\eta$-approximate dual feasibility and satisfy \Cref{lemma:setting_alpha}.

\textbf{For \boldmath$k=4$,}

\vspace{0.25cm}

\begin{tabular}{@{}ccccccc@{}}
\toprule
$x_{i_1}$         & $x_{i_2}$           & $\Delta P$                                        & Old $\alpha_{i_1}$ & Old $\alpha_{i_2}$       & New $\alpha_{i_1}=\alpha_{i_2}$ & $\beta_j$ \\ \midrule
$0$               & $1-\nicefrac{1}{8}$ & $1$                                               & $0$                & $\nicefrac{11}{4}-3\eta$ & $\eta$                          & $0$       \\
$\nicefrac{1}{2}$ & $1-\nicefrac{1}{4}$ & $\nicefrac{1}{2}+\nicefrac{1}{4}-\nicefrac{1}{8}$ & $1-\eta$           & $2-2\eta$                & $\eta$                          & $0$       \\ \bottomrule
\end{tabular}
\vspace{0.25cm}

\textbf{For \boldmath$k=5$,}
\vspace{0.25cm}

\begin{tabular}{ccccccc}
\hline
$x_{i_1}$         & $x_{i_2}$            & $\Delta P$                                                              & Old $\alpha_{i_1}$ & Old $\alpha_{i_2}$       & New $\alpha_{i_1}=\alpha_{i_2}$ & $\beta_j$ \\ \hline
$0$               & $1-\nicefrac{1}{16}$ & $1$                                                                     & $0$                & $\eta$                   & $\eta$                          & $0$       \\
$\nicefrac{1}{2}$ & $1-\nicefrac{1}{8}$  & $\nicefrac{1}{2}+\nicefrac{1}{8}-\nicefrac{1}{16}$                      & $1-\eta$           & $\nicefrac{11}{4}-3\eta$ & $\eta$                          & $0$       \\
$\nicefrac{1}{4}$ & $\nicefrac{1}{4}$    & $\nicefrac{1}{4}$+$\nicefrac{1}{4}-\nicefrac{1}{16}$ & $2-2\eta$          & $2-2\eta$                & $\eta$                          & $0$       \\ \bottomrule    
\end{tabular}
\vspace{0.25cm}

\textbf{For \boldmath$k=6$,}
\vspace{0.25cm}

\begin{tabular}{ccccccc}
\hline
$x_{i_1}$         & $x_{i_2}$            & $\Delta P$                                                              & Old $\alpha_{i_1}$ & Old $\alpha_{i_2}$       & New $\alpha_{i_1}=\alpha_{i_2}$ & $\beta_j$ \\ \hline
$0$               & $1-\nicefrac{1}{32}$ & $1$                                                                     & $0$                & $\eta$                   & $\eta$                          & $0$       \\
$\nicefrac{1}{2}$ & $1-\nicefrac{1}{16}$ & $\nicefrac{1}{2}+\nicefrac{1}{16}-\nicefrac{1}{32}$                     & $1-\eta$           & $\eta$                   & $\eta$                          & $0$       \\
$\nicefrac{1}{4}$ & $\nicefrac{1}{8}$    & $\nicefrac{1}{4}$+$\nicefrac{1}{8}\nicefrac{1}{32}$ & $2-2\eta$          & $\nicefrac{11}{4}-3\eta$ & $\eta$                          & $0$     
\\\bottomrule
\end{tabular}

\vspace{0.25cm}
\textbf{For \boldmath$k=7$,}
\vspace{0.25cm}

\begin{tabular}{ccccccc}
\hline
$x_{i_1}$         & $x_{i_2}$            & $\Delta P$                                          & Old $\alpha_{i_1}$       & Old $\alpha_{i_2}$       & New $\alpha_{i_1}=\alpha_{i_2}$ & $\beta_j$                  \\ \hline
$0$               & $1-\nicefrac{1}{64}$ & $1$                                                 & $0$                      & $\eta$                   & $\eta$                          & $0$                        \\
$\nicefrac{1}{2}$ & $1-\nicefrac{1}{32}$ & $\nicefrac{1}{2}+\nicefrac{1}{32}-\nicefrac{1}{64}$ & $1-\eta$                 & $\eta$                   & $\eta$                          & $0$                        \\
$\nicefrac{1}{4}$ & $\nicefrac{1}{16}$   & $\nicefrac{1}{4}+\nicefrac{1}{16}-\nicefrac{1}{64}$ & $2-2\eta$                & $\eta$                   & $\eta$                          & $0$                        \\
$\nicefrac{1}{8}$ & $\nicefrac{1}{8}$    & $\nicefrac{1}{8}+\nicefrac{1}{8}-\nicefrac{1}{64}$  & $\nicefrac{11}{4}-3\eta$ & $\nicefrac{11}{4}-3\eta$ & $\nicefrac{367}{64}-7\eta$      & $8\eta-\nicefrac{367}{64}$\\\bottomrule
\end{tabular}
\vspace{0.25cm}

For $T$, we can have three cases: (1) $T$ is of the form $2^m-1$ for some $m$. (2) $T$ is of the form $2^{m_1}+2^{m_2}-1$, or (3) $T$ is not of these forms.

\noindent{\bf Case 3:} We note that 
\begin{equation*}
\alpha_{(T)}=\min_{k\in \left\{0,1,2,\cdots,\frac{T-1}{2}\right\}}\left[\left(\frac{1}{2^{k}}+\frac{1}{2^{T-1-k}}-\frac{1}{2^{{T-1}}}\right)+\alpha_{(k)}+\alpha_{(T-1-k)}-\eta\right]
\end{equation*}
can be achieved by putting the entire $\Delta P$ into $\alpha$ without increasing the $\beta$. Now, considering the case we are in, we know that either $\alpha_{(k)}$ or $\alpha_{(T-1-k)}$ is $\eta$ (inductive hypothesis). Since the other $\alpha$ already had value more than $\eta-(1-x_i)$ (this can be seen from the fact that $\alpha_{(T)}$ converges to $\eta$ as $T$ tends to $\infty$ and the total value added to $\alpha_{(T)}$, at any point, is at most the increment in $x_i$) and observe that $\Delta P\geq (1-x_i)$. Therefore, the alpha can be increased to $\eta$.

\noindent{\bf Case 1:} In this case $T = 2^m-1$ for some $m$.
$$\alpha_{(2^m-1)}=\min_{k\in \left\{0,1,2,\cdots,2^{m-1}-1\right\}}\left(\frac{1}{2^{k}}+\frac{1}{2^{2^m-k-2}}-\frac{1}{2^{{2^m-2}}}\right)+\alpha_{(k)}+\alpha_{(2^m-k-2)}-\eta$$

Note that the minimum occurs when $k=2^{m-1}-1$ (on top of this, we need to show that $\Delta P-\Delta\alpha \geq 0$, so that $\beta$ can be set to a value in $[0,1]$). In this case,
\begin{align*}
    \alpha_{(2^m-1)}&=\left(\frac{1}{2^{2^{m-1}-1}}+\frac{1}{2^{2^m-2^{m-1}+1-2}}-\frac{1}{2^{{2^m-2}}}\right)+\alpha_{(2^{m-1}-1)}+\alpha_{(2^m-2^{m-1}+1-2)}-\eta\\
    & = \sum_{i=1}^{m-1} 2^{m-1-i} \left( \frac{1}{2^{2^i-2}} - \frac{1}{2^{2(2^i-2) + 2}}\right) + 2^{m-1} - (2^m -1) \eta.
\end{align*}
as required. Now, we need to show that the following quantity is non-negative.
\begin{align*}
\Delta P - \Delta \alpha &= \left(\frac{1}{2^{2^{m-1}-1}}+\frac{1}{2^{2^{m-1} -1 }}-\frac{1}{2^{{2^m-2}}}\right) - 2\cdot \left( \alpha_{(2^m-1)} - \alpha_{(2^{(m-1)}-1)}\right).\\
    & = \left(\frac{1}{2^{2^{m-1}-1}}+\frac{1}{2^{2^{m-1} -1 }}-\frac{1}{2^{{2^m-2}}}\right) - 2\cdot \left( \alpha_{(2^m-1)} - \alpha_{(2^{(m-1)}-1)}\right).\\
   & = \left(\frac{1}{2^{2^{m-1}-2}}-\frac{1}{2^{{2^m-2}}}\right) - 2\cdot \left( \alpha_{(2^m-1)} - \alpha_{(2^{(m-1)}-1)}\right).
\end{align*}

We have the following claim.
\begin{claim}\label{claim_case1}
The quantity 
    $\left(\frac{1}{2^{2^{m-1}-2}}-\frac{1}{2^{{2^m-2}}}\right) - 2\cdot \left( \alpha_{(2^m-1)} - \alpha_{(2^{(m-1)}-1)}\right)$ is non-negative for all values of $m\geq 1$.
\end{claim}
\begin{proof}

We first simplify the last term.
\begin{align*}
    &\left( \alpha_{(2^m-1)} - \alpha_{(2^{(m-1)}-1)}\right) \\
    &=  \left(\sum_{i=1}^{m-1} 2^{m-1-i} \left( \frac{1}{2^{2^i-2}} - \frac{1}{2^{2(2^i-2) + 2}}\right) + 2^{m-1} - (2^m -1) \eta \right) \\
    & \quad\quad\quad - \left(\sum_{i=1}^{m-2} 2^{m-2-i} \left( \frac{1}{2^{2^i-2}} - \frac{1}{2^{2(2^i-2) + 2}}\right) + 2^{m-2} - (2^{m-1} -1) \eta. \right)\\
    & = \frac{1}{2}\cdot \left( \frac{1}{2^{2^{m-1}-2}} - \frac{1}{2^{2(2^{m-1}-2) + 2}}\right)  + (2^{m-1} - 2^{m-2}) - (2^m - 2^{m-1})\eta\\
    & \quad\quad\quad + \sum_{i=1}^{m-1} (2^{m-1-i} - 2^{m-2-i})\left( \frac{1}{2^{2^i-2}} - \frac{1}{2^{2(2^i-2) + 2}}\right)\\
     & = \frac{1}{2}\cdot \left( \frac{1}{2^{2^{m-1}-2}} - \frac{1}{2^{2(2^{m-1}-2) + 2}}\right)  +  2^{m-2} - 2^{m-1}\left(1-\sum_{i=1}^\infty \frac{1}{2^{2^i+i-1}} \right)\\
    & \quad\quad\quad + \sum_{i=1}^{m-1}  2^{m-2-i}\left( \frac{1}{2^{2^i-2}} - \frac{1}{2^{2(2^i-2) + 2}}\right)\\
     & = \frac{1}{2}\cdot \left( \frac{1}{2^{2^{m-1}-2}} - \frac{1}{2^{2(2^{m-1}-2) + 2}}\right)  -  2^{m-2} + 2^{m-2}\sum_{i=1}^\infty \frac{2}{2^{2^i+i-1}} \\
    & \quad\quad\quad + 2^{m-2}\cdot \sum_{i=1}^{m-1}  \frac{1}{2^{i}}\cdot  \left(\frac{1}{2^{2^i-2}} - \frac{1}{2^{2(2^i-2) + 2}}\right).
\end{align*}
Now,
\begin{align*}
    2 \left( \alpha_{(2^m-1)} - \alpha_{(2^{(m-1)}-1)}\right) & = \left( \frac{1}{2^{2^{m-1}-2}} - \frac{1}{2^{2(2^{m-1}-2) + 2}}\right)  -  2^{m-1} + 2^{m-1}\sum_{i=1}^\infty \frac{2}{2^{2^i+i-1}} \\
    & \quad\quad\quad + 2^{m-1}\cdot \sum_{i=1}^{m-1}  \frac{1}{2^{i}}\cdot  \left(\frac{1}{2^{2^i-2}} - \frac{1}{2^{2(2^i-2) + 2}}\right)\\
     & = \left( \frac{1}{2^{2^{m-1}-2}} - \frac{1}{2^{2^{m}- 2}}\right)  -  2^{m-1} + 2^{m-1}\sum_{i=1}^\infty \frac{2}{2^{2^i+i-1}} \\
    & \quad\quad\quad + 2^{m-1}\cdot \sum_{i=1}^{m-1}  \frac{1}{2^{i}}\cdot  \left(\frac{1}{2^{2^i-2}} - \frac{1}{2^{2(2^i-2) + 2}}\right).\\
    \end{align*}
    Therefore,
\begin{align*}
    \Delta P - \Delta \alpha &=   2^{m-1} - 2^{m-1}\sum_{i=1}^\infty \frac{2}{2^{2^i+i-1}} - 2^{m-1}\cdot \sum_{i=1}^{m-1}  \frac{1}{2^{i}}\cdot  \left(\frac{1}{2^{2^i-2}} - \frac{1}{2^{2(2^i-2) + 2}}\right)\\
    &= 2^{m-1} \left( 1 - \sum_{i=1}^\infty \frac{2}{2^{2^i+i-1}} - \sum_{i=1}^{m-1}  \frac{1}{2^{i}}\cdot  \left(\frac{1}{2^{2^i-2}} - \frac{1}{2^{2(2^i-2) + 2}}\right)\right)\\
   &= 2^{m-1} \left(1 -  \sum_{i=1}^\infty \frac{2}{2^{2^i+i-1}} - \sum_{i=1}^{m-1}   \left(\frac{2}{2^{2^i+i-1}} - \frac{1}{2^{(2^{i+1}+i-2)}}\right) \right).
\end{align*}
Now, we simplify the last summation as follows
\begin{align*}
     \sum_{i=1}^{m-1}   \left(\frac{2}{2^{2^i+i-1}} - \frac{1}{2^{(2^{i+1}+i-2)}}\right) &=  \sum_{i=1}^{m-1}   \left(\frac{2}{2^{2^i+i-1}} - \frac{1}{2^{(2^{(i+1)}+(i+1)-3)}}\right)\\
     &=  \sum_{i=1}^{m-1}   \left(\frac{2}{2^{2^i+i-1}} - \frac{4}{2^{(2^{(i+1)}+(i+1)-1)}}\right)\\
     &= \frac{2}{2^2} - \sum_{i=2}^{m-1} \frac{2}{2^{2^i+i-1}} - \frac{4}{2^{(2^{m}+m-1)}}
\end{align*}
Hence,
\begin{align*}
    \Delta P - \Delta \alpha &=    2^{m-1} \left(1 -  \sum_{i=1}^\infty \frac{2}{2^{2^i+i-1}} - \sum_{i=1}^{m-1}   \left(\frac{2}{2^{2^i+i-1}} - \frac{1}{2^{(2^{i+1}+i-2)}}\right) \right)\\
    & =  2^{m-1} \left(1 -  \sum_{i=1}^\infty \frac{2}{2^{2^i+i-1}} - \frac{2}{2^2} + \sum_{i=2}^{m-1} \frac{2}{2^{2^i+i-1}} + \frac{4}{2^{(2^{m}+m-1)}}\right)\\
  & =  2^{m-1} \left(1 - \frac{2}{2^2} - \sum_{i=m}^\infty \frac{2}{2^{2^i+i-1}} - \frac{2}{2^2}  + \frac{4}{2^{(2^{m}+m-1)}} \right)\\
    & =  2^{m-1} \left(\frac{4}{2^{(2^{m}+m-1)}} - \sum_{i=m}^\infty \frac{2}{2^{2^i+i-1}} \right)\geq 0.
\end{align*}
\end{proof}

Given the claim, we can see that $\Delta P-\Delta \alpha\geq 0$, allowing us to set $\beta_j\geq 0$.
\end{proof}

\noindent{\bf Case 2:} In the case when $T$ is of the form $2^{m_1}+2^{m_2}-1$ (where $m_1>m_2$), we will try to find out $\alpha_{(T)}$ using $\alpha_{(2^{m_1}-1)}$ and $\alpha_{(2^{m_2}-1)}$. For all other instances in which we end up with $\alpha_{(T)}$, at least one of the offline vertices will already have $\alpha$ set to $\eta$ so like case $1$, we are in good shape. Hence we just need to consider the case when one offline vertex, say $i_1$, had $x_{i_1}$ at $1-\frac{1}{2^{2^{m_1}-1}}$ and the other vertex, $i_2$, had $x_{i_2}$ at $1-\frac{1}{2^{2^{m_2}-1}}$.

First of all, we note that
$$\Delta P=\frac{1}{2^{2^{m_1}-1}}+\frac{1}{2^{2^{m_2}-1}}-\frac{2}{2^{2^{m_1}+2^{m_2}-1}}.$$

The plan is to put all the $\Delta P$ into both the $\alpha$ values and confirm that indeed we can make both the $\alpha$ values equal to $\eta$, without incrementing the $\beta$.

This can be achieved iff $\Delta P$ is at least $2\eta -\alpha_{(2^{m_1}-1)}-\alpha_{(2^{m_2}-1)}$. The following claim precisely shows this.

\begin{claim}\label{claim_case2}
 $\Delta P\geq 2\eta -\alpha_{(2^{m_1}-1)}-\alpha_{(2^{m_2}-1)}$.
\end{claim}
\begin{proof}
We show that $\Delta P-2\eta+\alpha_{(2^{m_1}-1)}+ \alpha_{(2^{m_2}-1)}$ is non-negative as follows.
\begin{align*}
&\Delta P-2\eta+\alpha_{(2^{m_1}-1)}+ \alpha_{(2^{m_2}-1)} \\
&=\frac{1}{2^{2^{m_1}-1}}+\frac{1}{2^{2^{m_2}-1}}-\frac{2}{2^{2^{m_1}+2^{m_2}-1}}+\sum_{i=1}^{m_1-1}2^{m_1-1-i}\left(\frac{1}{2^{2^i-2}}-\frac{1}{2^{2(2^i-2)+2}}\right)+\\
&\quad\quad\quad\quad\sum_{i=1}^{m_2-1}2^{m_2-1-i}\left(\frac{1}{2^{2^i-2}}-\frac{1}{2^{2(2^i-2)+2}}\right)+2^{m_1-1}+2^{m_2-1}-(2^{m_1}+2^{m_2})\eta\\
& =\frac{1}{2^{2^{m_1}-1}}+\frac{1}{2^{2^{m_2}-1}}-\frac{2}{2^{2^{m_1}+2^{m_2}-1}}+2^{m_1-1}+2^{m_2-1}-(2^{m_1}+2^{m_2})\eta +\\
&\quad\quad\quad\quad 2^{m_1-1}\left(\frac{1}{2}-\sum_{i=2}^{m_1-1}\frac{2}{2^{2^i+i-1}}-\frac{4}{2^{(2^{m_1}+m_1-1)}}\right)+\\
&\quad\quad\quad\quad 2^{m_2-1}\left(\frac{1}{2}-\sum_{i=2}^{m_1-1}\frac{2}{2^{2^i+i-1}}-\frac{4}{2^{(2^{m_1}+m_1-1)}}\right)\\
&=-\frac{1}{2^{2^{m_1}-1}}-\frac{1}{2^{2^{m_2}-1}}-\frac{2}{2^{2^{m_1}+2^{m_2}-1}} + \\
&\quad\quad\quad\quad 2^{m_1-1}\left(1-2+2\sum_{i=1}^{\infty}\frac{1}{2^{2^i+i-1}}+\frac{1}{2}-\sum_{i=2}^{m_1-1}\frac{2}{2^{2^i+i-1}}\right)+\\
&\quad\quad\quad\quad 2^{m_2-1}\left(1-2+2\sum_{i=1}^{\infty}\frac{1}{2^{2^i+i-1}}+\frac{1}{2}-\sum_{i=2}^{m_2-1}\frac{2}{2^{2^i+i-1}}\right)\\
&=-\frac{1}{2^{2^{m_1}-1}}-\frac{1}{2^{2^{m_2}-1}}-\frac{2}{2^{2^{m_1}+2^{m_2}-1}}+2^{m_1-1}\left(\sum_{i=m_1}^{\infty}\frac{2}{2^{2^i+i-1}}\right)+\\
&\quad\quad\quad\quad 2^{m_2-1}\left(\sum_{i=m_2}^{\infty}\frac{2}{2^{2^i+i-1}}\right)\\
&\geq0
\end{align*}
\end{proof}

\section{Fractional Matching}\label{sec:frac}
In this section, we prove a lower bound on the competitive ratio of any algorithm for fractional matching and then give an online algorithm that achieves the best attainable competitive ratio.
\subsection{Lower Bound}
\begin{theorem}
    No online algorithm can attain a competitive ratio better than $0.75$ for online fractional bipartite matching.
\end{theorem}
\begin{proof}
    Consider a graph with two offline vertices $i_1$ and $i_2$. The first online vertex will have edges to both $i_1$ and $i_2$ while the second online vertex will have an edge to only one of the vertices. The optimal matching is of size $2$ but one can see that no online algorithm can attain a competitive ratio better than $0.75$ using this graph.
\end{proof}

\subsection{Upper Bound}
For the fractional matching, we shall consider the water-level algorithm which was an optimal algorithm for the case without the degree bound.

\begin{algorithm}[H]
  \caption{Water-Level Algorithm}\label{water}
  \begin{algorithmic}[1]

    \Procedure{$A$}{$j,S$}\Comment{Online vertex $j$ with neighborhood $S\subseteq L$, $|S|\leq 2$}
      \State Find $\ell$ such that $\sum_{i\in S}\max\{0,\ell -x_i\}$=1
      \If {$\ell\leq 1$} 
        \State Vertex $j$ can be completely matched.
        \State For all $i\in S$ with $x_i<\ell$, set $x_i=\ell$
        \Else
        \State All the neighbors of $j$ can be completely matched, but $j$ can only be partially matched.
        \State  For all $i\in S$ set $x_i=1$ 
    \EndIf
    \EndProcedure
  \end{algorithmic}
\end{algorithm}

We show that the same algorithm gives the optimal algorithm for the degree $2$ case.
\begin{theorem}
Water-level algorithm attains a competitive ratio of $0.75$ when the vertices have a degree at most $2$.
\end{theorem}
\begin{proof}
We shall use the primal-dual analysis for this. The main idea is to maintain the value of the fractional matching produced by water level and simultaneously maintain dual variables.

We already know how the primal variables are updated by water level, given the algorithm. We will now explain how the dual variables will be updated to maintain reverse weak duality and weak dual feasibility.

Let us define $x_i:=\sum_{j\in R}x_{i,j}$. The variable $x_i$ denotes how much the vertex $i$ is matched. We will maintain the $\alpha_i$'s such that, given $x_i$, we will know exactly what $\alpha_i$ is. Thus, we will know how much $\alpha_i$ increases by considering how much $x_i$ was incremented by water level. Given the increments in $\alpha_i$'s, we will know what value to set $\beta_j$ because we will maintain that the primal and dual objectives will always be the same. So, the dual objective will be increased by the same amount as the primal, which is the amount of matching done by the water level algorithm on the arrival of online vertex $j$.

Hence, first, we will describe how to maintain the $\alpha_i$'s.

\begin{tikzpicture}[scale=0.45,x=1cm,y=1cm]
\clip(-1.7,-1) rectangle (11,11);
\begin{axis}[
x=1cm,y=1cm,
axis lines=middle,
xmin=-0.5,
xmax=10.5,
ymin=-0.5,
ymax=10.5,
xtick={0,5,10},
ytick={0,2.5,...,10},yticklabel=\empty,xticklabel=\empty]
\clip(-2,-2) rectangle (12,12);
\end{axis}
\node at (5.4,-0.2) {0.5};
\node at (0,-0.1) {0};
\node at (10.4,-0.2) {1};
\node at (-0.3,2.9) {0.25};
\node at (-0.3,7.9) {0.75};
\node at (-0.2,10.4) {1};
\node at (3,-0.7) {$x_i \xrightarrow{\hspace{1cm}}$};
\node at (-1.3,4.4) {$\alpha_i$};
\node at (7.2,2.7) {$(0.5,0.25)$};
\node at (8.7,8) {$(1,0.75)$};
\node at (-1.3,6) {$\bigg\uparrow$};
\draw [line width=1pt] (0.5,0.5) -- (5,2.5);
\draw [line width=1pt] (5,2.5) -- (10,7.5);
\end{tikzpicture}

For the first half fraction of the matching, for any vertex $i\in L$, if a fractional matching of size $p$ is done between $i$ and $j\in R$, then we shall increment $\alpha_i$ by $p/2$ and $\beta_j$ by $p/2$. Whereas, if $u$ is already matched more than $0.5$, then for a matching of size $x$, we shall increment $\alpha_i$ by $x$, and we shall not increment $\beta_j$ at all.

Firstly, note that we never decrease any dual variables. So if  some dual constraint is satisfied at a given time then it will certainly be satisfied in the future. We shall now see that whenever a new vertex $j\in R$ arrives, both the newly introduced edges $(i_1,j)$ and $(i_2,j)$ satisfy approximate dual feasibility. This will complete the proof.

In order to observe the approximate dual feasibility, we shall look at a few cases based on how much $i_1$ and $i_2$ were matched right before the arrival of $j$.

\begin{enumerate}

\item {\bf $i_1$ and $i_2$ are completely unmatched.}

Notice that in this case, the water-level algorithm will match $j$ half to $i_1$ and a half to $i_2$, and the primal increases by $1$. Let us see how the dual variables are updated. Since $x_{i_1}$ and $x_{i_2}$ are both at $0.5$ we will have $\alpha_{i_1}=\alpha_{i_1}=0.25$ which adds $0.5$ to the dual objective. Hence, we can add $0.5$ to $\beta_j$.

Therefore $\alpha_{i_1}+\beta_j=\alpha_{i_2}+\beta_j=0.75$

\item {\bf $i_1$ is $x_{i_1}$ fraction matched and $i_2$ is completely unmatched}

We will assume that $i_1$ is $x_{i_1}$ fraction matched where $x_{i_1}\neq 0$. First, we note that the water-level algorithm maintains that if $x_{i_1}\neq 0$, then we must have $x_{i_1}\geq 0.5$. In other words, if a vertex $i\in L$ starts getting matched, it will be matched at least till the half-level mark. Hence we may assume $x_{i_1}\geq 0$.

In this case, water-level algorithm will first match $x_{i_1}$ fraction of vertex $i_2$ to $j$ and then $\frac{1-x_{i_1}}{2}$ to $i_1$ and $\frac{1-x_{i_1}}{2}$ to $i_2$. This will mean that $\frac{1+x_{i_1}}{2}$ will be the updated value of $x_{i_1}$ and $x_{i_2}$ after vertex $j$ is matched fractionally.

Hence, after the matching $\alpha_{i_1}=\alpha_{i_2}=0.25+\frac{x_{i_1}}{2}$ and $\beta_j=0.25$. Therefore $\alpha_{i_1}+\beta_j=\alpha_{i_2}+\beta_j=0.5+\frac{x_{i_1}}{2}\geq 0.75$ since $x_{i_1}$ was greater than $\frac{1}{2}$.

Note that the case where $i_1$ is completely unmatched but $i_2$ is partially matched is symmetric.

\item {\bf $i_1$ is $x_{i_1}$ fraction matched and $i_2$ is $x_{i_2}$ fraction matched}

In this last case, when both $i_1$ and $i_2$ were partially matched to begin with (hence, at least half matched as stated earlier), we can see that the water-level algorithm would end up matching $i_1$ and $i_2$ to $j$ - whatever fraction of $i_1$ and $i_2$ was still unmatched. After this, we will have $x_{i_1}=x_{i_2}=1$ and therefore $\alpha_{i_1}=\alpha_{i_2}=0.75$ and therefore $\alpha_{i_1}+\beta_j=\alpha_{i_2}+\beta_j=0.75$.
\end{enumerate}
\end{proof}

\section{Conclusion}
 An interesting open problem is finding the right competitive ratios for the integral online bipartite matching for (online) degree $d$ unweighted graphs.  If we denote the competitive ratios by $\mathcal{F}(d)$, for fractional, and $\mathcal{I}(d)$, for integral, respectively, then we have shown that $\mathcal{I}(2) = 0.717772\dots$. Given that $\mathcal{F}(2) = 0.75$, it would be interesting to find the values of  $\mathcal{I}(d)$ when $d\geq 3$ and compare them with $\mathcal{F}(d)$ especially given that we know both $\mathcal{I}(d)$ and $\mathcal{F}(d)$ approach \(1-\nicefrac1e\) as $d$ approaches infinity. We conjecture that $\mathcal{F}(d)$ and $\mathcal{I}(d)$ will be different for all $d\geq 2$.

It is also interesting to ask the more general question: given an algorithm that computes fractional solutions for a given class of graphs, how well can we round a fractional solution into a randomized integral solution?
% There is the more interesting question of rounding fractional solutions. We have shown that $G(n,2)$ is a collection of graphs where perfect rounding is not possible. What about other families of graphs? Can we improve on the known lower of bound of $\nicefrac{7}{8}$ for online rounding or instead achieve an online rounding of $\nicefrac{7}{8}$ for all fractional solutions? We hope to understand when rounding of a fractional solution is possible and to what extent.
% \newpage

% \section{Open Problems}

% \noindent Regarding this paper:
% \begin{enumerate}
%     \item {\color{cyan}Match the competitive ratio of an integral algorithm with the counterexample (0.7178) in the degree 2 case.}
%     \item Generalize the argument for higher degree graphs. 
% \end{enumerate}

% \noindent Regarding the OCS paper:
% \begin{enumerate}
%     \item Show that the OCS algorithm is optimal for every interval length. 
%     \item Generalize the 2-way OCS to a multi-way OCS algorithm.
% \end{enumerate}

% \noindent Regarding the integral matching result for general graphs
% \begin{enumerate}
%     \item Improve the analysis of Blanc-Charikar so that a wider range of OCSes fit into the primal-dual framework. 
% \end{enumerate}

{\small 
  \bibliographystyle{prahladhurl}
    \bibliography{degree2-bib}
}

\end{document}